\documentclass[10pt]{article}

\usepackage{graphicx,amsmath,amsfonts,amssymb,bm,hyperref,url,breakurl,epsfig,epsf,color,fullpage,MnSymbol,mathbbol,fmtcount,algorithmic,algorithm,semtrans,cite,caption,subcaption,multirow}
  
  \usepackage{mathtools}

\usepackage{titlesec}

\usepackage{tikz}
\usepackage{pgfplots}
\usetikzlibrary{pgfplots.groupplots}

\titleformat{\paragraph}
{\normalfont\normalsize\bfseries}{\theparagraph}{1em}{}
\titlespacing*{\paragraph}
{0pt}{3.25ex plus 1ex minus .2ex}{1.5ex plus .2ex}

\usepackage{movie15}

\usepackage{cite}
\usepackage{caption}
\usepackage[bottom,hang,flushmargin]{footmisc} 

\usepackage{hyperref}
\definecolor{darkred}{RGB}{150,0,0}
\definecolor{darkgreen}{RGB}{0,150,0}
\definecolor{darkblue}{RGB}{0,0,200}
\hypersetup{colorlinks=true, linkcolor=darkred, citecolor=darkgreen, urlcolor=darkblue}

\newtheorem{theorem}{Theorem}[section]
\newtheorem{lemma}[theorem]{Lemma}
\newtheorem{corollary}[theorem]{Corollary}

\newtheorem{definition}[theorem]{Definition}
\newtheorem{condition}[theorem]{Condition}

\newcommand\tr{{{\operatorname{trace}}}}
\newcommand{\eps}{\varepsilon}

\newcommand{\beq}{\begin{equation}}
\newcommand{\eeq}{\end{equation}}

\newcommand{\nn}{\nonumber}

\newcommand{\A}{{\mtx{A}}}
\newcommand{\Cb}{{\mtx{C}}}
\newcommand{\Ub}{{\mtx{U}}}
\newcommand{\V}{{\mtx{V}}}
\newcommand{\B}{{{\mtx{B}}}}
\newcommand{\Gb}{{\mtx{G}}}
\newcommand{\F}{{\mtx{F}}}

\newcommand{\Pb}{{\mtx{P}}}

\newcommand{\Rb}{{\mtx{R}}}

\newcommand{\Xii}{\X_{r_i}}
\newcommand{\Yi}{\Y_{r_i}}
\newcommand{\Iden}{{\bf{I}}}
\newcommand{\M}{{\mtx{M}}}
\newcommand{\Vb}{{\mtx{V}}}

\newcommand{\Pc}{\mathcal{P}}
\newcommand{\Sc}{\mathcal{S}}

\newcommand{\Nn}{\mathcal{N}}

\newcommand{\vb}{\mtx{v}}
\newcommand{\rb}{\mtx{r}}
\newcommand{\pb}{\mtx{p}}
\newcommand{\qb}{\mtx{q}}
\newcommand{\w}{\mtx{w}}

\newcommand{\ab}{\vct{a}}

\newcommand{\bb}{\vct{b}}
\newcommand{\ub}{\vct{u}}
\newcommand{\deltab}{\delta_{buff}}
\newcommand{\epsb}{\eps_{buff}}
\newcommand{\h}{\vct{h}}

\newcommand{\g}{\vct{g}}

\newcommand{\infnorm}[1]{\left\|#1\right\|_{\ell_\infty}}
\newcommand{\twonorm}[1]{\left\|#1\right\|_{\ell_2}}

\newcommand{\x}{{\vct{x}}}
\newcommand{\X}{{\mtx{X}}}
\newcommand{\Hb}{{\mtx{H}}}
\newcommand{\Y}{{\mtx{Y}}}
\newcommand{\Yp}{{\mtx{Y}^\perp}}
\newcommand{\y}{{\vct{y}}}

\definecolor{emmanuel}{RGB}{255,127,0}

\newcommand{\R}{\mathbb{R}}
\newcommand{\Pro}{\mathbb{P}}

\newcommand{\sgn}[1]{\textrm{sgn}(#1)}
\renewcommand{\P}{\operatorname{\mathbb{P}}}
\newcommand{\E}{\operatorname{\mathbb{E}}}

\newcommand{\e}{\mathrm{e}}

\newcommand{\diag}[1]{\text{diag}{(#1)}}
\newcommand{\vct}[1]{\bm{#1}}
\newcommand{\mtx}[1]{\bm{#1}}

\newcommand{\order}[1]{{\cal{O}}( #1)}
\newcommand{\tn}{\twonorm}

\newcommand{\ino}{\infnorm}

\numberwithin{equation}{section} 

\def \endprf{\hfill {\vrule height6pt width6pt depth0pt}\medskip}
\newenvironment{proof}{\noindent {\bf Proof} }{\endprf\par}

\newcommand{\ang}{\text{ang}}

\title{Near-Optimal Sample Complexity Bounds for\\Circulant Binary Embedding} 
\author{Samet Oymak\thanks{Google Inc.~1600 Amphitheatre Parkway, Mountain View, CA 94043.}}
\begin{document}
\maketitle

\begin{abstract} Binary embedding is the problem of mapping points from a high-dimensional space to a Hamming cube in lower dimension while preserving pairwise distances. An efficient way to accomplish this is to make use of fast embedding techniques involving Fourier transform e.g.~circulant matrices. While binary embedding has been studied extensively, theoretical results on fast binary embedding are rather limited. In this work, we build upon the recent literature to obtain significantly better dependencies on the problem parameters. A set of $N$ points in $\R^n$ can be properly embedded into the Hamming cube $\{\pm 1\}^k$ with $\delta$ distortion, by using $k\sim\delta^{-3}\log N$ samples which is optimal in the number of points $N$ and compares well with the optimal distortion dependency $\delta^{-2}$. Our optimal embedding result applies in the regime $\log N\lesssim n^{1/3}$. Furthermore, if the looser condition $\log N\lesssim \sqrt{n}$ holds, we show that all but an arbitrarily small fraction of the points can be optimally embedded. We believe our techniques can be useful to obtain improved guarantees for other nonlinear embedding problems.
\end{abstract}

\section{Introduction}
Binary embedding problem aims to map a set of points in a high-dimensional space to the Hamming cube in a lower dimension. The task is preserving the distances between the points while keeping embedding dimension as small as possible. A common approach to accomplish this task is applying a random map to the data. In particular, given a point $\x\in\R^n$, we first apply a linear transformation $\x\rightarrow\A\x\in\R^k$ and then apply the discretization $\A\x\rightarrow \sgn{\A\x}$ where $\sgn{\cdot}$ returns the sign. Given a set $S$ and distortion level $\delta>0$, we are interested in ensuring that for all $\x,\y\in S$, $\A$ satisfies
\beq
|k^{-1}\|\sgn{\A\x},\sgn{\A\y}\|_H-\ang(\x,\y)|\leq \delta.\nn
\eeq
Here, $\|\cdot,\cdot\|_H$ is the Hamming distance between two $\{0,1\}^k$ vectors and $\ang(\cdot)$ is the angular distance which returns the smaller angle between two points normalized by $\pi$. Often we are interested in embedding a large set of points $S=\{\vb_i\}_{i=1}^N$ or a continuous set such as a subspace. An important aspect of the embedding problems is the tradeoff between the number of points $N$ and the embedding dimension $m$. For linear embedding, classical Johnson-Lindenstrauss (JL) Lemma guarantees that by using $k\approx \delta^{-2}\log N$ samples, $N$ points can be embedded with $\delta$ distortion. More recently, this tradeoff attracted significant attention for the binary embedding problem. Specifically, by choosing $\A$ to be a Gaussian matrix, it can be trivially shown that one can achieve a good binary embedding under the same assumption of $k\approx \delta^{-2}\log N$. This arguments have also been extended to arbitrary (e.g.~continuous) sets which are of interest for sparse estimation problems. 

While the results on dense Gaussian matrices are valuable, for most applications we are interested in faster projections where embedding can be done in near-linear time. Such projections make use of fast matrix multiplications such as the Fourier Transform followed by random diagonal modulations and are broadly called Fast Johnson-Lindenstrauss Transform (FJLT). In this work, we focus on circulant embedding matrices where projection matrix $\A$ is given by $\A=\Rb\Cb_{\h}\diag{\rb}$. Here,
\begin{itemize}
\item $\Rb\in\R^{k\times n}$ is the restriction operator that selects $k$ rows out of $n$ uniformly at random.
\item $\h,\rb\in\R^n$ are independent vectors with independent standard normal entries.
\item $\Cb_{\h}$ is a circulant matrix whose first row is equal to $\h^*$.
\item $\diag{\rb}$ is a diagonal matrix obtained from the vector $\rb$.
\end{itemize}
The theoretical results for fast binary embedding techniques are rather limited \cite{yi2015binary,yu2015binary,oymak2015near}. Related to us, very recently Yu et al. provided an analysis of circulant projections. Loosely speaking, the authors show that by using $k\sim \log^2 N$ samples, binary embedding with small distortion is possible as long as $\log N \lesssim n^{1/6}$. Another related work connected to nonlinear embedding is due to Le et al. \cite{le2013fastfood}. Here, the authors speed up Kernel approximation \cite{rahimi2007random} by making use of FJLT however the number of required Fourier features scale quadratically due to suboptimal concentration bounds. A natural question is whether circulant projections can achieve the optimal binary embedding guarantees. In this work, we answer this question positively. We show that using $k\sim \log N$ samples, binary embedding via circulant matrices will be successful as long as $\log N\lesssim n^{1/2}$. This shows that Fast JL Transform not only works well for linear embedding but also for highly nonlinear problems and the embedding behavior is essentially same.

\noindent{\bf{Contributions:}} Specifically, we have two sets of results. Our first set of results consider embedding with circulant projections and the associated theorem has a dependency on the coherence of the set $\{\vb_i\}_{i=1}^N$. When the points are not spiky, (i.e.~small infinity norm), the optimal embedding works for a larger regime of $N$. For maximally incoherent sets we can allow $\log N \lesssim n^{1/2}$. Our second result is a corollary of the first one and attempts to remove the dependence on incoherence. This is done by applying an additional layer of randomness $\x\rightarrow \Hb\diag{\bb}\vb$ where $\Hb$ is the Hadamard transform and $\diag{\bb}$ is a diagonal matrix with independent Rademacher diagonal entries. The overall embedding takes the form $\vb\rightarrow \sgn{\Rb\Cb_{\h}\diag{\rb}\Hb\diag{\bb}\vb}$. Observe that all matrix multiplications are still near-linear time. This model makes no assumption on the set $\{\vb_i\}_{i=1}^N$ and optimal embedding is possible as soon as $\log N\lesssim n^{1/3}$. Furthermore, if $\log N\lesssim \sqrt{n}$, fast and optimal binary embedding still succeeds for all but arbitrarily small fraction of the points.


\subsection{Related Literature}
Binary embedding with dense Gaussian matrices is a rather well studied problem. Guarantees for finite set of points can be obtained by applying a standard Chernoff bound. Embedding continuous sets is a more challenging problem and it is studied in a series of papers \cite{plan2014dimension,oymak2015near,boufounos20081,plan2013robust,jacques2013robust} with results mostly restricted to Gaussian ensemble. Much less is known for the fast embedding techniques that make use of Fourier or Hadamard transforms. We can split the existing results in this direction into two groups.

\noindent $\bullet$ {\bf{Fast JL embedding followed by dense (two-layer) map:}} This map is given by $\x\rightarrow \sgn{\Gb\F\diag{\bb}\x}$ where $\Gb\in\R^{k\times k'}$ is a dense Gaussian matrix, $\F\in\R^{k'\times n}$ is the subsampled Discrete Fourier Transform matrix, $\h$ has independent Rademacher entries, and $k\approx k'\approx \order{\log N}$. This first applies a fast linear dimensionality reduction to linearly embed points to the lower dimensionality space $\R^{k'}$. Next, we use a dense Gaussian matrix to obtain a binary embedding guarantee. This approach is not computationally efficient as soon as $k\gtrsim \order{\sqrt{n}}$ as dense Gaussian multiplication becomes more expensive than the Fourier transform. In \cite{yi2015binary}, Yi et al. propose a related but more efficient algorithm by replacing $\Gb$ with a more involved procedure involving Toeplitz matrices.

The optimal embedding bound of the present paper applies in the regime $k\leq \sqrt{n}$ which shows that circulant projections perform as good as two-layer maps computationally (both require $\order{n\log n}$ in the regime $k\leq \order{\sqrt{n}}$). However, the proposed approach is much simpler and easily extends to the regime $k\geq \order{\sqrt{n}}$ in an efficient manner (albeit without proof).


\noindent $\bullet$ {\bf{Simply use Fast JL embedding:}} We simply apply a Fast JL Transform by using a circulant matrix. The map we consider has the form $\x\rightarrow\sgn{\Rb\Cb_{\h} \diag{\rb}}$ where $\Rb\in\R^{k\times n}$ is the subsampling operator, $\Cb_{\h}$ is a circulant matrix and $\rb,\h$ are vectors with iid $\Nn(0,1)$ entries. Since circulant matrices are diagonalized by the Discrete Fourier Transform, the computational complexity of embedding is always $\order{n\log n}$ independent of the sample size $k\leq n$. Yu et al. \cite{yu2015binary} very recently provided an analysis of this map with rigorous sample complexity bounds. While their result has significant dependency on the set geometry, under best circumstances, they show that $k\gtrsim \log^2 N$ samples are sufficient for successful embedding as long as $\log N\lesssim n^{1/6}$. As an example of geometric dependence, the results of \cite{yu2015binary} depends on the maximal correlation of the point set $\sup_{i\neq j}|\vb_i^*\vb_j|$ and becomes arbitrarily suboptimal as this number approaches $1$. In particular, they cannot allow points that are close to each other. There are also several works on the applications of fast binary projections in large scale image retrieval and hashing algorithms \cite{yu2014circulant,wang2014binary,gong2011iterative}.

\section{Main results}
Suppose we are given $N$ unit vectors in $\R^n$ namely $\{\vb_i\}_{i=1}^N$. Our task is mapping this points to a low-dimensional Hamming cube in $\R^k$ while preserving the distances. We are interested in ensuring that for all $1\leq i,j\leq N$, $\A$ satisfies
\beq
|k^{-1}\|\sgn{\A\vb_i},\sgn{\A\vb_j}\|_H-\ang(\vb_i,\vb_j)|\leq \delta.\nn
\eeq
As a geometric feature, we shall make use of the coherence of the set which is defined as \[\rho_{cross}=\max\{\sup_{1\leq i\leq N}\ino{\vb_i}, ~\sup_{1\leq i\neq j\leq N} \frac{\ino{\vb_i-\vb_j}}{\tn{\vb_i-\vb_j}}\}.\] For our results to work, we make the following assumptions on $N,k,n$ and the coherence parameter.
\begin{condition} \label{cond1}There exists sufficiently large nonnegative constants $c_1,c_2,c_3$\footnote{$c,C,\{c_i,C_i\}_{i\geq 0},c',C'$ will be used to denote absolute constants that may vary from line to line.}, such that
\begin{enumerate}
\item $k>c_1\delta^{-3}\log N$.
\item $c_2\delta k\rho_{cross} \log n < 1$.
\item $\delta \geq c_3k\rho_{cross} $.
\end{enumerate}
\end{condition}
Observe that in the maximally incoherent case ($\rho_{cross}=\order{n^{-1/2}}$), we can pick $\delta=o(1)$, $k=\order{(\log n)^{-1}n^{1/2}}$ and $\log N=\order{\delta^{-3} k}$. Hence, our optimal embedding result applies up to $\order{\sqrt{n}}$ as the embedding dimension. Our main result is on fast binary embedding of finite set of points with near-optimal embedding dimensions and is stated in the next theorem.
\begin{theorem}\label{main result} Let $\A=\Rb\Cb_{\h}\diag{\rb}\in\R^k$ be a circulant projection as described above. Under the assumptions of Condition \ref{cond1}, with probability $1-\exp(-c_4\delta^3k)$, for all $\x,\y\in\{\vb_i\}_{i=1}^N$, we have that
\beq
|k^{-1}\|\sgn{\A\x},\sgn{\A\y}\|_H-\ang(\x,\y)|\leq \delta.\nn
\eeq
\end{theorem}


This result applies to arbitrary set of points; however, it depends on the incoherence of the set $\rho_{cross}$. One can get rid of this dependency by applying an additional layer of randomization. In particular, let $\Hb$ be a Hadamard matrix of size $n$ and let $\bb\in\R^n$ be a vector with independent Rademacher entries. If $n$ is not a power of $2$, we can simply zero-pad the vectors. Consider the map \[\A_{\Hb}=\A\Hb\diag{\bb}=\Rb\Cb_{\h}\diag{\rb}\Hb\diag{\bb}.\] For this map, we have the following result that is incoherence-free.
\begin{theorem} \label{no assumption}Consider the binary embedding via the operator $\x\rightarrow \sgn{\A_{\Hb}\x}$. There exists universal constants $c,C>0$ such that following holds. Suppose 
\[\log N\leq c\delta^{2}(\log n)^{-1}n^{1/3}.\] Then, with probability $1-\exp(-c\log N)$, the point set $\w_i=\Hb\diag{\bb}\vb_i$ obeys the incoherence condition with $\rho_{cross}=C\delta (\log n)^{-1/2}n^{-1/3}$. Consequently, as soon as $k\geq c_1\delta^{-3}\log N$, with probability $1-\exp(-c_4\delta^3k)$,
\beq
|k^{-1}\|\sgn{\A_{\Hb}\x},\sgn{\A_{\Hb}\y}\|_H-\ang(\x,\y)|\leq \delta.\nn
\eeq 
\end{theorem}
\begin{proof} This result follows from the fact that the set of points obtained by the map $\vb_i\rightarrow \Hb\bb\vb_i$ has desirable geometric features (small $\rho_{cross}$) with high probability. In particular, combine Theorem \ref{main result} with Lemma \ref{rand modulation}.
\end{proof}
Finally, the next result shows that one can optimally embed most of the points as long as $\log N\lesssim \order{\sqrt{n}}$.
\begin{theorem}\label{no assumption some}Consider the binary embedding via the operator $\x\rightarrow \sgn{\A_{\Hb}\x}$. There exists universal constants $c,C>0$ such that following holds. Suppose 
\[\log N\leq c\delta^{3}(\log n)^{-2}n^{1/2}.\] Then, with probability $1-n^{-2}$ (over $\Hb$), there exists $S_{good}\subseteq \{\vb_i\}_{i=1}^N$ such that
\[|S_{good}|\geq (1-c_5n^{-2})N,~~~\text{and}~~~\text{for all }\vb\in S_{good}~\text{:}~\ino{\Hb\diag{\bb}\vb}\leq \rho_{cross}\]
where $\rho_{cross}=C \sqrt{\log n/n}$. Consequently, as soon as $k\geq c_1\delta^{-3}\log N$, with probability $1-n^{-2}-\exp(-c_4\delta^3k)$, all $\x,\y$ chosen from $S_{good}$ obeys
\beq
|k^{-1}\|\sgn{\A_{\Hb}\x},\sgn{\A_{\Hb}\y}\|_H-\ang(\x,\y)|\leq \delta.\nn
\eeq 
\end{theorem}
\begin{proof} This result follows from the fact that all but a small fraction of the set of points obtained by the map $\vb_i\rightarrow \Hb\bb\vb_i$ has desirable geometric features (small $\rho_{cross}$) with high probability. In particular, combine Theorem \ref{main result} with Lemma \ref{most modulation}. Pick $p=n^{-2}$ in Lemma \ref{most modulation}.
\end{proof}

\section{Conclusions and Open Problems}
In this work, we showed that fast binary embedding with near optimal dimensions are possible. In particular, our embedding bounds are consistent with the state of the art results for linear embedding, indicating that fast binary embedding is feasible under identical conditions to fast linear embedding such as \cite{ailon2006approximate,krahmer2011new,oymak2015isometric}. This is the first such result for fast binary embedding and significantly improves over related literature (e.g. \cite{yu2015binary,le2013fastfood}). We believe the tools developed in this paper broadly applies to nonlinear embedding tasks. For instance, our argument may be used to improve the concentration estimates of Fastfood features \cite{le2013fastfood} which is a popular fast kernel approximation technique. Our embedding result holds for finite set of points and it is of interest to extend this work to continuous sets. A weakness of our result is the fact that the embedding dimension scales up to $\order{\sqrt{n}}$ which limits the number of points to $\log N\lesssim \order{\sqrt{n}}$. This work opens up several research directions.
\begin{itemize}
\item {\bf{Fast embedding in linear regime:}} Does fast binary embedding work with embedding dimension $n$? In other words, can we pick $k\sim \order{n}$ to embed $N\sim \exp(\order{k})$ points? If not, is there a fundamental bottleneck at $k\sim\order{\sqrt{n}}$?
\item {\bf{Practical considerations:}} Our result on circulant embedding $\Cb_{\h}\diag{\rb}$ requires $\h$ and $\rb$ to have Gaussian entries. We believe $\rb$ can have Rademacher entries without impacting the performance. It would possibly improve the performance as the operator $\vb\rightarrow \diag{\rb}\vb$ preserves the inner products when $\rb$ is Rademacher. Furthermore, it is not clear whether the incoherence assumption in Theorem \ref{main result} is necessary. Numerical results of prior work \cite{yu2015binary,yu2014circulant} indicates that the map $\vb\rightarrow\text{sign}(\Cb_{\h}\diag{\rb}\vb)$ works well which suggests that we may not need additional randomization via Hadamard transform. This would allow us to discard one layer of the embedding, namely, $\vb\rightarrow \Hb\diag{\bb}\vb$.
\item {\bf{General nonlinear embedding:}} With a minor modification of our analysis, it is possible to obtain fast embedding bounds for a more general model $f(\A\x)$ where $f$ is a function that apply pointwise. The important use cases would be to replace $\sgn{\cdot}$ function with a general function such as quantization, ReLU, sigmoid etc \cite{giryes2015deep,jacques2015quantized}. It would also be of interest to investigate quadratic samples arising in phase retrieval \cite{jaganathan2013sparse,candes2015phase}. 
\item {\bf{Embedding of continuous sets:}} Our current results apply to finite set of points however it is of interest to embed continuous sets such as subspaces or sparse and low-rank manifolds. While this problem is studied for dense Gaussian embedding matrices, we believe similar results can be obtained for fast embedding matrices by building on this work and \cite{yu2015binary}.
\end{itemize}

The rest of the paper is dedicated to the proof of our main result Theorem \ref{main result}. Before going into technical details, we introduce the necessary notation. Given a vector $\x\in\R^n$, let $s_i(\x)$ be the vector obtained by shifting entries of $\x$ by $i$ position, i.e. $s_i(\x)$. In particular, $j$th entry of $s_i(x)$ is same as $(i+j)$th entry of $\x$ modulo $n$. $\sigma_{\min}(\cdot)$ and $\sigma_{\max}(\cdot)$ returns minimum and maximum singular values of a matrix respectively. $\|\cdot\|$ denotes the spectral norm of a matrix and is same as $\sigma_{\max}(\cdot)$. $\text{diag}(\cdot)$ returns a diagonal matrix from a vector input or returns the vector of diagonal entries of a matrix. $c,C,\{c_i,C_i\}_{i\geq 0},c',C'$ will be used to denote absolute constants. For nonzero $\x$, $\bar{\x}=\x/\tn{\x}$ and $\bar{0}=0$. Throughout this work, Hadamard and Discrete Fourier Transform matrices are normalized to be unitary. A standard Gaussian vector obeys the distribution $\Nn(0,\Iden)$. $S$ denotes a subset of $\{1,2,\dots, n\}$ obtained by picking $k$-elements uniformly at random without replacement. Define direct coherence to be $\rho_{direct}=\sup_{1\leq i\leq N} \ino{\vb_i}$. Let $\theta>0$ be the smallest angle between these points namely $\min_{i\neq j} \ang(\vb_i,\vb_j)$. It is trivial to show that the cross coherence can be bounded as $\rho_{cross}\leq 2\sin(\theta)^{-1}\rho_{direct}$. $\ell$th entry of a vector of size $n$ is same as ``$\ell~(\text{mod}~n)$''th entry of the vector.


\section{Controlling the Conditioning of the Projection Matrix}

To simplify our notation, we shall assume that $N\geq n$. $n>N$ case can be recovered by setting $n=N$ in our main result. Let $\{r_i\}_{i=1}^k$ be distinct numbers selected from the set $\{1,2,\dots,n\}$ uniformly at random.
\begin{definition} [Random shift vectors] Let $\x\in\Sc^{n-1}$ and let $\rb\in\R^n$ be a standard Gaussian vector. Random shift vectors of $\x$ are a set of random vectors $\{\X_i\}_{i=1}^n$ such that $\X_i=s_i(\diag{\rb}\x)$ for $0\leq i\leq n-1$. Define $\Y_i$ in the identical manner given vector $\y$ for the same choice of $\rb$.
\end{definition}

The following theorem summarizes the main result of this section by providing a spectral norm bound on subsampled random shift vectors.
\begin{theorem}\label{thm main} Pick unit vectors $\x,\y\in\R^n$ satisfying $\x^*\y=0$ and $\ino{\x},\ino{\y}\leq \rho$. Form a matrix $\M\in\R^{n\times 2k}$ by picking the same $k$ vectors $\{\Xii\}_{i=1}^k,\{\Yi\}_{i=1}^k$ from each of $\{\X_i\}_{i=1}^n$ and $\{\Y_i\}_{i=1}^n$ without replacement uniformly at random and then stacking next to each other. With probability $1-2\exp(-\delta^2k)$ (over $\rb$ and selection of $\{\Xii\}_{i=1}^k$'s), we have that
\beq
\sigma_{\max}(\M^*\M-\Iden)\leq C\delta k\rho\log n.
\eeq
\end{theorem}

\begin{corollary} \label{cor main} Let $\x,\y$ be unit vectors obeying $\ang(\x,\y)=\theta$ and $\ino{\x},\ino{\y}\leq \rho$. Form a matrix $\M\in\R^{n\times 2k}$ by picking the same $k$ vectors $\{\Xii\}_{i=1}^k,\{\Yi\}_{i=1}^k$ from each of $\{\X_i\}_{i=1}^n$ and $\{\Y_i\}_{i=1}^n$ without replacement uniformly at random and then stacking next to each other. With probability $1-6\exp(-\delta^2k)$ (over $\rb$ and selection of $\{\Xii\}_{i=1}^k$'s), we have that
\beq
\sigma_{\max}(\M^*\M-\Iden_\theta)\leq C\delta k\rho\log n.
\eeq
where $\Iden_\theta\in\R^{2k\times 2k}$ is given by the matrix $\begin{bmatrix}\Iden_k&\cos(\theta)\Iden_k\\\cos(\theta)\Iden_k&\Iden_k\end{bmatrix}$.
\end{corollary}
\begin{proof} Proof is based on Theorem \ref{thm main}. Consider the decomposition $\y=\cos(\theta)\x+\sin(\theta)\y'$ where $\x^*\y'=0$. Denote the $k$ chosen columns $\{\Xii\}_{i=1}^k$ by the matrix $\X\in\R^{n\times k}$ and the corresponding matrix to $\{\Yi\}_{i=1}^k$ by $\Y\in\R^{n\times k}$ and set $\Y'=\sin(\theta)^{-1}(\Y-\cos(\theta)\X)$. Now observe that
\beq
\X^*\Y=\cos(\theta)\X^*\X+\sin(\theta)\X^*\Y'.\nn
\eeq
Using Theorem \ref{thm main} on $\X^*\X$, we know that for an absolute constant $c_1>0$, with probability $1-2\exp(-\delta^2k)$
\beq
\|\X^*\X-\Iden\|\leq c_1k\rho\log n.\nn
\eeq
Next we apply Theorem \ref{thm main} to the matrix $[\X~\Y']$. From Lemma \ref{ortho cost}, $\y'$ obeys $\ino{\y'}\leq \frac{2\rho}{\sin(\theta)}$. Consequently, we have the spectral norm estimate
\beq
\|\X^*\Y'\|\leq \|[\X~\Y']^*[\X~\Y']\|\leq  \sin(\theta)^{-1}c_1\delta k\rho\log n.\nn
\eeq
Combining these, and using triangle inequality, we obtain
\beq
\|\X^*\Y-\cos(\theta)\Iden\|\leq \sin(\theta)\|\X^*\Y'\|+\cos(\theta)\|\X^*\X-\Iden\|\leq 2c_1\delta k\rho\log n+\cos(\theta)c_1\delta k\rho\log n\leq 3c_1\delta k\rho\log n.\nn
\eeq
Finally, we need to estimate the remaining submatrices. In particular, direct applications of Theorem \ref{thm main} yields
\beq
\max\{\|\X^*\X-\Iden\|,\|\Y^*\Y-\Iden\|\}\leq c_1\delta k\rho\log n.\nn
\eeq
Combining these estimates and representing $\M$ as $4$ $k\times k$ submatrix involving $\X^*\X,\Y^*\Y,\X^*\Y,\Y^*\X$, we find
\beq
\|\M^*\M-\Iden_\theta\|\leq \|\Y^*\Y-\Iden\|+\|\X^*\X-\Iden\|+2\|\X^*\Y-\cos(\theta)\Iden\|\leq 8c_1\delta k\rho\log n.\nn
\eeq
Using a union bound, this final event happens with probability $1-6\exp(-\delta^2k)$.
\end{proof}

\subsection{Proof of Theorem \ref{thm main}}
\begin{theorem}[Hanson-Wright Theorem \cite{rudelson2013hanson}] Let $\A\in\R^{n\times n}$ and $\g\in\R^n$ be a standard Gaussian vector. There exists a constant $c>0$ such that
\beq
\Pro(|\g^*\A\g-\E[\g^*\A\g]|\geq t)\leq 2\exp\left(-c\left\{\frac{t^2}{\|\A\|_F^2},\frac{t}{\|\A\|}\right\}\right).\nn
\eeq
\end{theorem}
The following lemma follows as a corollary of Hanson-Wright Theorem.
\begin{lemma} \label{hs app}Let $\x,\y$ be two unit vectors where $\ino{\x},\ino{\y}\leq \rho$ and $\x^*\y=0$. Let $\g$ be a standard Gaussian vector and $\Gb=\diag{\g}$. Then, the followings hold
\item For all $1\leq i\neq j\leq n$
\beq
\max\{\Pro(|s_i(\Gb\x)^*s_j(\Gb\y)|>t),\Pro(|s_i(\Gb\x)^*s_j(\Gb\x)|>t),\Pro(|\tn{\Gb\x}^2-1|>t)\}\leq 2\exp(-c\min\{\frac{t}{\rho^2},\frac{t^2}{\rho^2}\}).\nn
\eeq
If $\ino{\x}\leq \rho,\ino{\y}\leq \rho'$ we additionally have $\Pro(|s_i(\Gb\x)^*s_j(\Gb\y)|>t)\leq  2\exp(-c\min\{\frac{t}{\rho\rho'},\frac{t^2}{\rho^2}\})$.
\end{lemma}
\begin{proof} The proofs are based on Hanson-Wright Theorem. $s_i(\Gb\x)^*s_j(\Gb\y)$ can be viewed as $\g^*\M\g$ where $\M$ is a weighted permutation matrix whose $(\ell+i,\ell+j)$th entry is of the form $\x_{\ell+i}\y_{\ell+j}$ for $1\leq \ell\leq n$ and whose remaining entries are $0$. Consequently, this matrix has maximum spectral norm $\rho^2$ and maximum Frobenius norm of $\rho$.  $\E[s_i(\Gb\x)^*s_j(\Gb\y)]$ is clearly $0$ when $i\neq j$ and for $i=j$ it is equal to $\x^*\y=0$. Hence, Hanson-Wright yields the desired bound.

For the second and third relations, identical argument applies. We additionally use the fact that $\E[\tn{\Gb\x}^2]=1$. The last statement follows by modifying the spectral norm estimate from $\rho^2$ to $\rho\rho'$.
\end{proof}

\begin{theorem}\label{lemma main} Pick unit vectors satisfying $\x^*\y=0$, $\ino{\x},\ino{\y}\leq \rho$, and $\rho<c_0\max\{\delta,(\log n)^{-1/2}\}$ for a sufficiently small constant $c_0>0$. Define the matrix $\A=[\X_1~\dots~\X_n~\Y_1~\dots~\Y_n]$. Form the matrix $\M\in\R^{n\times 2k}$ by picking the same $k$ columns from $\{\X_i\}_{i=1}^n$ and $\{\Y_i\}_{i=1}^n$ uniformly at random. Denote the indices of these columns (i.e.~support set) by $S$ which is a subset of $\{1,2,\dots,n\}$. Let $\Phi\in\R^{n\times 2n}$ be the matrix obtained by normalizing the columns of $\A$. Let $\Pb$ be the matrix obtained by normalizing the columns of $\M$ which is a submatrix of $\Phi$. With probability $1-3e^{-\delta^2k}$ over the generation of $S$, we have that
\beq
\E[\sigma_{\max}(\M^*\M-\Iden)]\leq ck\rho\log n.
\eeq
\end{theorem}
\begin{proof} We first calculate the coherence of the matrix $\Phi$ which is defined as $\mu(\Phi)=\sup_{i\neq j}|\phi_i^*\phi_j|$ where $\phi_i$ is the $i$th column of $\Phi$ for $1\leq i\leq 2n$.
\begin{lemma} $\Pro(\mu(\Phi)\leq c_1\sqrt{\log n}\rho)\geq 1- 4n^{-3}$ where the probability is over $\rb$ and $c_1>0$ is an absolute constant.
\end{lemma}
\begin{proof} For some absolute constant $c_1=8c_2^{-1/2},c_2>0$ we have the followings. Applying Lemma \ref{hs app}, for all terms, we have that $|\tn{\X_i}^2-1|,|\X_i^*\X_j|,|\X_i^*\Y_j|,|\X_i^*\Y_i|\leq \gamma$ with probability $1-4n^2\exp(-c_2\min\{\gamma\rho^{-2},\gamma^2\rho^{-2}\})$. Hence, picking $\gamma=c_1\sqrt{\log n}\rho/2$, we can guarantee that $\sup_{i\neq j}|\phi_i^*\phi_j|$ is small for all $i,j$ pairs with probability $1-4n^{-3}$ after normalizing the columns by their lengths $\tn{\X_i},\tn{\Y_i}$. 
\end{proof}
Next, Lemma \ref{lemma markov} shows that the spectral norm of $\Phi$ can be bounded as $\|\Phi\|\leq c_3\rho\sqrt{n \log n}$ with probability $1-n^{-3}$ as well. Assume $c_3>c_1$ without losing generality.

Let us call the event that ``$\mu(\Phi)\leq c_3\sqrt{\log n}\rho$ and $|\tn{\X_i}^2-1|\leq c_3\sqrt{ \log n}\rho$ and $\|\Phi\|\leq c_3\rho\sqrt{n \log n}$'' as $E$ which is an event over $\rb$ with probability at least $1-5n^{-3}$. Split $\P$ into $\X$ and $\Y$ parts namely $\Pb=[\Pb_\X~\Pb_\Y]$ where $\Pb_\X,\Pb_\Y\in\R^{n\times k}$. Conditioned on $E$, applying Theorem \ref{tropp var} with $u=2\delta\sqrt{k}\log k$, with probability $1-e^{-\delta^2 k}$ over the choice of support $S$, we have that
\begin{align}
&\E_{\rb\big|E}[\sigma_{\max}(\Pb_\X^*\Pb_\X-\Iden)]\leq c_4k\rho\log n
&\E_{\rb\big|E}[\sigma_{\max}(\Pb_\Y^*\Pb_\Y-\Iden)]\leq c_4k\rho\log n\nn
\end{align}
since we can bound \eqref{rand var} as follows
\begin{align}
\mu(\Phi) \sqrt{k}u+\frac{k}{n}\|\Phi\|^2&\leq c_3(\rho\sqrt{\log n}\sqrt{k}\sqrt{k}\sqrt{\log k}\delta+\frac{k}{n}(\rho^2n\log n))\nn\\
&\leq c_3(k\delta\rho\log n+k\rho^2\log n)\leq 2c_3k\delta\rho\log n.\nn
\end{align}
where we used the fact that $\delta\geq \rho$. Similarly, applying Theorem \ref{tropp var2}, we estimate the cross term as
\[
\E_{\rb\big|E}[\sigma_{\max}(\Pb_\X^*\Pb_\Y)]\leq c_4k\rho\log n
\]
This yields
\begin{equation}
\E_{\rb\big|E}[\sigma_{\max}(\Pb^*\Pb-\Iden)]\leq4c_4k\rho\log n\nn
\end{equation}
with probability $1-3e^{-\delta^2 k}$. Next, observe that
\beq
\Pb=\M\diag{\alpha}\nn
\eeq
where $\alpha\in\R^{2k}$ is a vector whose entries lie between $\sqrt{1\pm c_3\sqrt{\log n}\rho}$. Let $c_5=\max\{2c_3,8c_4\}$. Consequently, applying Lemma \ref{lemma simple but useful}
\begin{equation}
\sigma_{\max}(\M^*\M-\Iden)\leq c_5k\rho\log n.\label{e bound}
\end{equation}

For the complementary event $\bar{E}$, independent of the support $S$ we will use a simpler estimate namely 
\beq
\sigma_{\max}(\M^*\M-\Iden)\leq \sigma_{\max}(\M^*\M)\leq \|\M\|_F^2=k(\tn{\X_1}^2+\tn{\Y_1}^2)\leq 2k\rho^2\tn{\rb}^2.\nn
\eeq
For this case, applying Lemma \ref{gauss tail} with $p=5n^{-3}$ yields that
\beq
\E[\tn{\rb}^2\big| \bar E]\Pro(\bar {E})\leq c_6n^{-2}\label{eb bound}
\eeq
 Combining the estimates over $E$ \eqref{e bound} and $\bar E$ \eqref{eb bound}, we find that with the desired probability over $S$ ($1-e^{-\delta^2k}$),
\begin{align}
\E_{\rb}[\sigma_{\max}(\M^*\M-\Iden)]&=\E_{\rb}[\sigma_{\max}(\M^*\M-\Iden)\big|E]\Pro(E)\nn\\
&+\E_{\rb}[\sigma_{\max}(\M^*\M-\Iden)\big|\bar{E}]\Pro(\bar{E})\nn\\
&\leq c_5k\rho\log n+c_6n^{-2}\approx c_5k\rho\log n\nn
\end{align}
where we used the fact that $k\rho\log n\geq n^{-2}$.
\end{proof}

\subsection{Probabilistic bounds on the singular values}
\begin{lemma} Let $\Rb\in\R^{n\times k}$ be a matrix obtained by picking $k$ elements from $\{\X_i\}_{i=1}^n\subset\R^n$ and stacking them next to each other. The maximum and minimum singular values of $\Rb$ are $\sqrt{k}\ino{\x}$ Lipschitz function of $\rb$.
\end{lemma}
\begin{proof} We view $\Rb$ as a random matrix obtained from the vector $\rb$. Given an alternative vector $\hat\rb$, construct $\hat\Rb$ from circular shifts of the vector $\diag{\hat\rb}\x$ in an identical manner to $\Rb$ (i.e. form $\{\hat{\X}_i\}_{i=1}^n$ and pick the same $k$ elements). Applying Lemma \ref{inf lip}, we have that
\beq
\|\hat\Rb-\Rb\|_F=\sqrt{k}\tn{\diag{\hat\rb}\x-\diag{\rb}\x}\leq \sqrt{k}\ino{\x}\tn{\rb-\hat\rb}\nn
\eeq
which is the desired conclusion.
\end{proof}

\begin{lemma} \label{lemma markov}Let $\X=[\X_1~\dots~\X_n]$ where $\X_i$ are obtained by circular shifts of $\diag{\rb}\x$. There exists an absolute constant $c>0$ such that
\beq
\Pro(\|\X\|\geq c\rho\sqrt{n\log n})\leq c_1n^{-3}.\label{advert 1}
\eeq
Next, consider the matrix $\Phi$ of Theorem \ref{lemma main}. Assuming $\rho<c'(\log n)^{-1/2}$, $\Phi$ obeys the following similar bound
\beq
\Pro(\|\Phi\|\geq 4c\rho\sqrt{n\log n})\leq (2c_1+4)n^{-3}.\nn
\eeq
\end{lemma}
\begin{proof} Let $g\sim\Nn(0,1)$. From Stirling's approximation, we have that
\beq
\E[g^{2d}]= (2d)!!=\frac{(2d)!}{2^d d!}\leq \frac{e(2d)^{2d+1/2}\exp(-2d)}{2^d\sqrt{2\pi}d^{d+1/2}\exp(-d)}\leq c_1(2/e)^dd^d.\nn
\eeq
Pick a complex standard normal $g'=g_1+ig_2$ where $g_1,g_2\sim\Nn(0,1)$. Comparing the moments of $g'$ to $g$
\beq
\E[|g'|^{2d}]=\E[(g_1^2+g_2^2)^d]\leq \E[(g_1^2+g_1^2)^d]\leq c_1(4/e)^dd^d.\nn
\eeq
Suppose $h$ is a real random variable obeying $\E[h^{2d}]\leq n\E[|g'|^{2d}]\leq c_1n(4/e)^dd^{d}$ for all $d\geq 1$. Then, using Markov inequality
\beq
\Pro(|h|\geq t)\leq t^{-2d}c_1n(4/e)^dd^{d}.\nn
\eeq
Pick $t=\sqrt{2d}$ and $d=(C\log n)/2$ to find that \[\Pro(|h|\geq t)\leq t^{-2d}c_1n(2/e)^dd^{d}\leq c_1n (2/e)^{(C\log n)/2}\leq c_1n^{-3}\] by picking $C>0$ to be a large enough constant. This gives 
\beq
\Pro(|h|\geq \sqrt{C\log n})\leq c_1n^{-3}.\label{s norm est}
\eeq
The remaining discussion will analyze the spectral norm of $\X$ to make use of \eqref{s norm est}. Observe that the random variable $\text{tr}((\X^*\X)^d)\geq \|\X\|^{2d}$. Now form the complex circulant matrix $\X'$ by stacking circular shifts $\{s_i(\diag{\g}\x')\}_{i=1}^n$ next to each other where are entries of the vector $\x'$ are equal to $\rho$ i.e. $\x'=[\rho~\dots~\rho]^*$ and $\g$ is a vector of independent random variables where each entry is distributed as $g'$.

Singular values of $\rho^{-1}\X'$ are trivial. In particular, singular values are absolute values of the eigenvalues and eigenvalues are independent complex Gaussian random variables whose imaginary and real parts have variance $\sqrt{n}$. Next, we relate $\X'$ to $\X$. Let $\X'_{real}$ denote the real part of the matrix $\X'$. First observe that, the following deterministic relation holds for all $d\geq 1$
\[\tr({\X'}^*\X')^d\geq \tr((\X'_{real})^*\X'_{real})^d.\]

On the other hand, $\E\tr((\X'_{real})^*\X'_{real})^d\geq \E\tr(\X^*\X)^d$. This follows form the fact that when the traces are expanded term by term, each individual nonzero term of the left-hand side dominates that of the right-hand side as entries of $\x'$ are at least as large as that of $\x$ (in absolute value). Finally observe that $\tr({\X'}^*\X')^d\sim \sum_{i=1}^n g_i^{2d}$ where $g_i$ are independent standard complex with variance $2\rho^2n$. Consequently, setting $h=\|\X\|$, we find that
\beq
\E[h^{2d}]=\E[\|\X\|^{2d}]\leq \E\tr(\X^*\X)^d\leq \E\tr({\X'}^*\X)^d=\E[\sum_{i=1}^n g_i^{2d}].\nn
\eeq
To conclude with the proof of the first statement, apply the estimate \eqref{s norm est} by normalizing both sides by $\rho\sqrt{n}$ and obtain the advertised result \eqref{advert 1}. 

To obtain the result on $\Phi$, recalling Theorem \ref{lemma main}, we write $\Phi\diag{\alpha}=[\X~\Y]$ where $\X,\Y=[\Y_1~\dots~\Y_n]$ have spectral norm at most $c\rho\sqrt{n\log n}$ and $\alpha$ is a diagonal (length) normalization matrix whose entries are at least $1/2$ with probability $1-4n^{-3}$ as soon as $\rho<c'(\log n)^{-1/2}$ for sufficiently small constant $c'>0$.
\end{proof}

\subsubsection{Finalizing the Proof of Theorem \ref{thm main}}
\begin{lemma} Consider the setup in Theorem \ref{thm main} and set $\M=[\X_{r_1}~\dots~\X_{r_k}~\Y_{r_1}~\dots~\Y_{r_k}]$. There exists a constant $c_1>0$ such that with probability $1-4\exp(-\delta^2k)$ over the generation of $\{r_i\}_{i=1}^k$ and modulation $\rb$, we have that
\beq
\|\M^*\M-\Iden\|\leq c_1\rho \delta k\log n.\nn
\eeq
\end{lemma}
\begin{proof} From Theorem \ref{lemma main}, we know that with probability $1-3e^{-\delta^2k}$ over support $S$
\beq
\E[\|\M^*\M-\Iden\|]\leq c\delta k\rho\log n.\nn
\eeq
On the other hand, $\sigma_{\min}(\M)$ and $\sigma_{\max}(\M)$ are $\sqrt{k}\ino{\x}$ Lipschitz functions of $\rb$. Consequently, conditioned on $S$, applying Lemma \ref{lip gauss}, we have that
\beq
\Pro(|\|\M\|- \E[\|\M\|]|\leq \delta \sqrt{2}k\ino{\x})\leq 2\exp(-\delta^2k).\nn
\eeq
Combining the expectation and deviation estimates, with probability $1-2\exp(-\delta^2k)$, we obtain that
\beq
\|\M\|\leq \sqrt{1+c\delta k\rho\log n}+\sqrt{2}\delta k\rho\leq 1+(c+\sqrt{2})\delta k\rho\log n .\nn
\eeq
The exact same argument applies to the minimum singular value $\sigma_{\min}(\M)$ which gives
\beq
\sigma_{\min}(\M)\geq \sqrt{1-c\delta k\rho\log n}-\sqrt{2}\delta k\rho\geq 1-(c+\sqrt{2})\delta k\rho\log n\nn
\eeq
allowing us to conclude with the desired result.
\end{proof}

\section{On orthogonal decomposition of Gaussian circulant pairs}\label{ortho sec}
Let $\x,\y$ be two unit vectors chosen from $\{\vb_i\}_{i=1}^N$. Form $\{\Xii,\Yi\}_{i=1}^k$ via uniform sampling of Gaussian circular rotations $\{\X_i,\Y_i\}_{i=1}^n$. Decompose $\Xii=\Xii'+\pb_i$, $\Yi=\Yi'+\pb'_i$ where $\pb_i$, $\pb'_i$ are the projections of $\Xii,\Yi$ onto the span of $\{\X'_{r_j},\Y'_{r_j}\}_{j=1}^{i-1}$. Observe that this has a similar flavor to QR decomposition.

\begin{lemma} \label{p individual}Let $S_i$ be the subspace spanned by $\{\X_{r_j},\Y_{r_j}\}_{j=1}^{i-1}$. With probability $1-4\exp(-\delta^2 k)$, we have that for all $1\leq i\leq k$
\beq
\max\{\tn{\Pc_{S_i}(\Xii)}=\tn{\pb_i},~\tn{\Pc_{S_i}(\Yi)}=\tn{\pb'_i}\}\leq c_1\delta k\rho_{direct}.
\eeq
\end{lemma}
\begin{proof} Our proof is in similar spirit to Lemma $12$ of \cite{yu2015binary}. The main difference is that we apply an additional orthogonalization procedure that reduces dependency on the correlation $|\x^*\y|$ and improves our estimates. To start analysis, let us focus on $\Xii$ only. First observe that 
\beq
\text{span}([\X_{r_1}~\dots~\Xii~\Y_{r_1}~\dots~\Yi])=\text{span}([\X'_{r_1}~\dots~\X'_{r_i}~\Y'_{r_1}~\dots~\Y'_{r_i}]).\nn
\eeq 
Next observe that
\beq
\text{span}([\X_{r_1}~\dots~\Xii~\Y_{r_1}~\dots~\Yi])=\text{span}([\X_{r_1}~\dots~\Xii~\Y^\perp_{r_1}~\dots~\Y^\perp_{r_i}])\nn
\eeq
where $\Yp_{r_1}$ is obtained by the Gaussian circular rotations of $\y^\perp=\frac{\y-\cos(\theta)\x}{\tn{\y-\cos(\theta)\x}}$ where $\theta$ is the angle between $\x$ and $\y$. Consequently, we can focus on understanding the projection of $\Xii$ onto the column span of $\M_i=[\X_{r_1}~\dots~\X_{r_{i-1}}~\Yp_{r_1}~\dots~\Yp_{r_{i-1}}]$. Let $\M_i$ have singular value decomposition $\Ub_L\Sigma\Ub_R^*$ where $\Sigma\in\R^{2(i-1)\times 2(i-1)}$. Consider the vector
\beq
\qb_i=\M_i^*\Xii\in\R^{2(i-1)}\nn
\eeq
From Lemma \ref{hs app}, we know that each entry of $\qb_i$ is less than $c_2\max\{\rho_{direct}\rho_{cross}\delta^2k,\rho_{direct}\delta\sqrt{k}\}\leq c_2\delta \rho_{direct}\sqrt{k}$ for all $1\leq i\leq k$ with probability $1-\exp(-\delta^2k)$ where we used the fact that $\delta^2k\geq c'\log N\geq c'\log n$. On the other hand, using Theorem \ref{thm main}, with the same probability all matrices $\{\M_i\}_{i=1}^k$ satisfy
\beq
\|\M_i^*\M_i-\Iden\|\leq c_3\rho_{cross} k\log n\implies \sigma_{\min}(\M_i)=\sigma_{\min}(\Sigma)\geq 1- c_3\rho_{cross} \delta k\log n\geq 1/2.\nn
\eeq
Consequently, the projection can be bounded as
\beq
\tn{\Pc_{S_i}(\Xii)}=\tn{\Ub_L^*\Xii}\leq \sigma_{\min}^{-1}(\Sigma)\tn{\Sigma\Ub_L^*\Xii}=\sigma_{\min}^{-1}(\Sigma)\tn{\M_i^*\Xii}.\nn
\eeq
This implies that, with $1-2\exp(-\delta^2k)$ probability, $\tn{\Pc_{S_i}(\Xii)}\leq c_2\sqrt{2k\times (\delta \rho_{direct}\sqrt{k})^2}=2c_2\delta k\rho_{direct}$. The identical argument applies to $\Yi$.
\end{proof}

\begin{lemma}\label{p spectral} Consider the matrix $\Pb\in\R^{n\times {2k}}$ obtained by concatenating $\pb_i,\pb_i'$ for $1\leq i\leq k$. Under initial assumptions, we have that $\|\Pb\|\leq 7$ with probability $1-8\exp(-\delta^2k)$.
\end{lemma}
\begin{proof} Consider the matrix $\M=[\X_{r_1}~\dots~\X_{r_k}~\Y_{r_1}~\dots~\Y_{r_k}]$. From Corollary \ref{cor main}, we know that $\|\M\|\leq 3$ with probability $1-6\exp(-\delta^2k)$. On the other hand, using Gaussian concentration, each column of $\M$ obeys
\beq
\Pro(\tn{\Xii}\leq 2)\leq 1- \exp(-0.5\rho^{-2}).\nn
\eeq
Using our initial assumption $\delta\geq c'k\rho$ (see Condition \ref{cond1}), this holds for all columns with probability $1-k\exp(-0.5\rho^{-2})\geq 1-\exp(-\delta^2k)$. Given this, observe that $\Xii'$ is perpendicular to $\{\X_{r_{j}}'\}_{j\neq i}$ and $\tn{\Xii'}\leq \tn{\Xii}\leq 2$. This ensures that
\beq
\|[\X_{r_1}'~\dots~\X_{r_k}']\|\leq \max_{1\leq i\leq k}\tn{\X_{r_k}}\leq 2.\nn
\eeq
The same argument applies to $\Yi'$ ensuring $\M'=[\X_{r_1}'~\dots~\X_{r_k}'~\Y_{r_1}'~\dots~\Y_{r_k}']$ has spectral norm of at most $4$. Consequently $\|\Pb\|=\|\M-\M'\|\leq \|\M\|+\|\M'\|\leq 7$.
\end{proof}

\begin{lemma} The matrix $\Pb=[\pb_1~\dots~\pb_k~\pb'_1~\dots~\pb'_k]$ obeys the following bounds with probability $1-12\exp(-c\delta^2k)$.
\begin{itemize}
\item Each column $\pb_i,\pb'_i$ of $\Pb$ satisfies $\tn{\pb_i},\tn{\pb'_i}\leq C\delta \rho k$ for all $1\leq i\leq k$.
\item Spectral norm of $\Pb$ satisfies $\|\Pb\|\leq 7$.
\end{itemize}
\end{lemma}
\begin{proof} The proof follows directly by making use of Lemmas \ref{p spectral} and \ref{p individual}.
\end{proof}

\section{Final perturbation analysis}

We are in a position to prove our main result Theorem \ref{main result}.

\begin{proof} The proof is based on perturbation analysis, namely to what extent structured samples deviate from Gaussian-like behavior. We break the analysis in two parts, namely over $\rb$ and over $\h$.\\
\noindent {\bf{$\bullet$~Upper bounds on the perturbation due to $\rb$:}}\\
Recall that $\Cb=\Cb_{\h}$ is the circulant part of the embedding operator where $\h\sim\Nn(0,\Iden)$ is its first row and $i$th row is equal to $s_i(\h)$ for $1\leq i\leq n$. Given any two points $\x,\y$, chosen from $\{\vb_i\}_{i=1}^N$ consider the vectors $\x'=\diag{\rb}\x=\X_1$ and $\y'=\diag{\rb}\y=\Y_1$. Now, observe that the $i$th entry of $\Cb\x'$ is equal to 
\beq
s_i(\h)^*\x'=\h^*s_{n-i}(\x')=\h^*\X_{n-i}.\nn
\eeq
Similarly the $i$th entry of $\Cb\y'$ is equal to $\h^*\Y_{n-i}$. Consequently, for a random subsampling $\Rb\Cb\in\R^{k\times n}$ of $\Cb$, we have that
\beq
\Rb\Cb\x'=\M_\x\h,~\Rb\Cb\y'=\M_\y\h,\nn
\eeq
where $\M_\x=[\X_{r_1}~\dots~\X_{r_k}]^*$ and $\M_\y=[\Y_{r_1}~\dots~\Y_{r_k}]^*$ and $\{r_i\}_{i=1}^k$ are randomly selected coordinates. Next, for each $1\leq i\leq k$, we decompose $\X_{r_i},\Y_{r_i}$ as described in Section \ref{ortho sec}.
\beq
\X_{r_i}=\X_{r_i}'+\pb_i,~\Y_{r_i}=\Y_{r_i}'+\pb'_i.\nn
\eeq
Since $\h$ is a standard Gaussian vector, by construction, $\h^*\Xii'$ and $\h^*\Yi'$ is independent of $\{\h^*\X_{r_j}',\h^*\Y_{r_j}'\}_{j\neq i}$. To proceed, let us estimate the angle between $\Xii',\Yi'$ probabilistically.

Firstly, $\tn{\Xii}^2,\tn{\Yi}^2$ lies between $1\pm c_1\delta\rho\sqrt{k}$ with probability $1-\exp(-\delta^2k)$. Next, with the same probability $|\Xii^*\Yi-\x^*\y|\leq c_1\delta\rho\sqrt{k}$. Together these imply that $|\bar\X^*_{r_i}\bar\Y_{r_i}-\x^*\y|\leq 4c_1\delta\rho\sqrt{k}$ where $\bar{\ab}=\ab/\tn{\ab}$. Making use of Lemma \ref{angular perturb}, we can conclude that
\beq
|\ang(\Xii,\Yi)-\ang(\x,\y)|\leq c_2\sqrt{\rho\delta\sqrt{k}}:=\Delta_{\rb}.
\eeq
In particular, since $\Xii,\Yi$ are circulant rotations of $\X_{r_1},\Y_{r_1}$ the angle between is exactly same i.e. $\ang(\X_{r_i},\Y_{r_i})=\ang(\X_{r_j},\Y_{r_j})$ for $1\leq i,j\leq k$.

With these, we can state the following result that summarizes the properties of the perturbation. Below we additionally used the fact that $\log N\leq \delta^2 k/4$.
\begin{lemma} $\{\Xii,\Yi\}$ satisfies the following with probability $1-12\exp(-\delta^2 k/2)$ for all $\x,\y$ pairs chosen from $\{\vb_i\}_{i=1}^N$ where the probability is over $\rb$ and support $S$.
\begin{itemize}
\item $|\ang(\Xii,\Yi)-\ang(\x,\y)|\leq \Delta_{\rb}$.
\item $\{\Xii',\Yi'\}_{i=1}^k$ are orthogonal pairs and for all $i$, $\Xii-\Xii'=\pb_i$, $\Yi-\Yi'=\pb_i'$ where $\pb_i,\pb_i'$ obey
\beq
\tn{\pb_i}\leq C\rho\delta k,~\|[\pb_1~\dots~\pb_k]\|\leq 7.\nn
\eeq
\end{itemize}
\end{lemma}

What remains is to characterize the effect of perturbation error on the binary embedding distortion. Let $\theta_i=\ang(\Xii',\Yi')$. Applying Lemma \ref{angular perturb} again, we know for a fact that (by picking $c_2>0$ to be a large enough constant)
\beq
\ang(\Xii',\Xii)\leq c_2\delta \rho k/2,~\ang(\Yi',\Yi)\leq c_2\rho \delta k/2.\nn
\eeq
Together, these ensure that
\beq
|\theta_i-\ang(\x,\y)|=|\ang(\Xii',\Yi')-\ang(\x,\y)|\leq c_2\delta\rho k+|\ang(\Xii,\Yi)-\ang(\x,\y)|\leq c_2\delta\rho k+\Delta_{\rb}:=\Delta'_{\rb}.\nn
\eeq
$\Delta'_{\rb}$ will be the source of embedding distortion due to $\rb$ and our initial assumptions will guarantee that it is small. Next section develops estimates for the remaining source of the perturbation which is connected to $\h$.

\noindent{\bf{$\bullet$~Probabilistic analysis of the perturbation due to $\h$:}}\\
Pick $\deltab>0$. For the rest of the discussion probabilities will be over $\h$. Let us define the events
\begin{align}
&E_i=(\h^*\Xii'>\deltab~\text{and}~\h^*\Yi'<-\deltab)~\text{or}~(\h^*\Xii'<-\deltab~\text{and}~\h^*\Yi'>\deltab),\nn\\
&\bar{E}_i=(\h^*\Xii'>\deltab~\text{and}~\h^*\Yi'>\deltab)~\text{or}~(\h^*\Xii'<-\deltab~\text{and}~\h^*\Yi'<-\deltab).\nn
\end{align}
$E_i$ and $\bar E_i$ are the robust versions of the events $\sgn{\h^*\Xii'}\ne\sgn{\h^*\Yi'}$ and $\sgn{\h^*\Xii'}=\sgn{\h^*\Yi'}$ respectively.

Without losing generality, let us consider the event $E_i$. Recall that with probability $1-\exp(-\delta^2k)$, for all $1\leq i\leq k$, we can guarantee that $0.75\leq\tn{\Xii}^2\leq 2$ and $\tn{\pb_i}^2\leq 0.25$. Hence, conditioned on $\rb$, $\h^*\Xii'$ (and $\h^*\Xii'$) is a Gaussian random variable with variance between $0.5$ to $2$. Also, observe that $\Pro(\sgn{\h^*\Xii'}\ne\sgn{\h^*\Yi'})=\theta_i$. Consequently, letting $\theta=\ang(\x,\y)$, from small ball probability of Gaussians, we have that
\beq
\Pro(E_i)\geq \ang(\theta_i)-c_3\deltab\geq \ang(\theta)-c_3(\deltab+\Delta'_{\rb}).\nn
\eeq
Let $E=\sum_{i=1}^k1_{E_i}$. Consequently, applying a standard Chernoff bound, we find that
\beq
\Pro(E\geq k(\ang(\theta)-c_3(\deltab+\Delta'_{\rb})-\deltab)):=\Pro(E\geq k(\ang(\theta)-c_4(\deltab+\Delta'_{\rb})))\geq \exp(-2\deltab^2k)\label{lower bad}
\eeq
where the probability is over $\h$.

Next, we consider the impact of perturbations $\{\pb_i,\pb_i'\}_{i=1}^k$. Using the facts that $\|\Pb\|_F^2\leq C^2\rho^2\delta^2k^3, \|\Pb\|\leq 7$ and applying Lemma \ref{expo bound}, we have that
\beq
\Pro(\tn{[\pb_1~\pb_2~\dots~\pb_k]^*\h}^2\leq \rho^2 \delta^2k^3+t)\geq 1-\exp(-c_5\min\{\frac{t^2}{50C^2\rho^2\delta^2k^3},\frac{t}{50}\}).\nn
\eeq
To proceed, pick $t=\eps_{buff}\deltab^2 k$ to obtain that with probability $1-\exp(-c_6\min\{\frac{\eps_{buff}^2\deltab^4 }{\rho^2\delta^2k},\epsb\deltab^2 k\})$, perturbation obeys
\beq
\tn{\Pb^*\h}^2\leq \rho^2\delta^2 k^3+\epsb\deltab^2 k.\label{bad est}
\eeq
The same bound applies to the perturbation over $\y$ namely $\Pb'=[\pb'_1~\dots~\pb'_k]$. Now for any $1\leq i\leq k$, observe that $\sgn{\Xii}\ne \sgn{\Yi}$ whenever
\beq
\text{i)~$E_i$ holds ~~and}~~~\text{ii)}~\max\{|\h^*(\Xii-\Xii')|,|\h^*(\Yi-\Yi')|\}<\deltab.\label{bad event}
\eeq
We know that $E_i$ holds on at least $k(\ang(\theta)-c_4(\deltab+\Delta'_{\rb}))$ coordinates. Next, we can upper bound the number of coordinates for which \eqref{bad event} does not hold. Using the estimate \eqref{bad est}, this number is given by
\beq
\frac{\tn{\Pb^*\h}^2+\tn{\Pb'^*\h}^2}{\deltab^2}\leq k\epsb+\deltab^{-2}\rho^2 \delta^2k^3.\label{upper bad}
\eeq
With the estimates \eqref{upper bad} and \eqref{lower bad}, we find that for all pairs $\x,\y$ with probability \[1-N^2\exp(-c_6\min\{\frac{\eps_{buff}^2\deltab^4 }{\rho^2\delta^2k},\epsb\deltab^2 k\})-\exp(-\delta^2k/2)\] we have that
\beq
k^{-1}\|\Rb\A\x,\Rb\A\y\|_H\geq \ang(\x,\y)-[ \epsb+\deltab^{-2}\delta^2\rho^2 k^2+c_3(\deltab+\Delta'_{\rb})].\nn
\eeq
The identical (symmetric) argument allows us to obtain the upper bound on the Hamming distance to conclude that
\beq
|k^{-1}\|\Rb\A\x,\Rb\A\y\|_H-\ang(\x,\y)|\leq  \epsb+\deltab^{-2}\delta^2\rho^2 k^2+c_3(\deltab+\Delta'_{\rb}).\nn
\eeq
With these bounds, we find that binary embedding with $c_{final}\delta$ distortion succeeds with probability $1-2\exp(-c_6\delta^3k/2)$ under the following conditions:
\begin{itemize}
\item $\epsb\leq \delta$,
\item $\deltab\leq \delta$,
\item $\rho^2k^2\leq \deltab^2\delta^{-1}$,
\item Via $\Delta'_{\rb}$: $\rho\delta\sqrt{k}\leq \delta^2$ i.e. $\rho\sqrt{k}\leq \delta$,
\item Via $\Delta'_{\rb}$: $\rho \delta k\leq \delta$.
\end{itemize}
To satisfy these, pick $\deltab=\epsb=\delta$. Furthermore, our initial assumptions (Condition \ref{cond1}) guarantee that $\delta\geq C_0\rho k\geq  C_0\max\{\rho^2 k^2,\rho\sqrt{k}\}$ for a sufficiently large constant $C_0>0$ which yields a total distortion proportional to $\delta$. Finally, the probability of success is 
\[1-N^2\exp(-c_6\min\{\frac{\delta^4 }{\rho^2k},\delta^3 k\})-\exp(-\delta^2k/2).\]
Observing $\rho^2k\leq \delta^2/C_0^2$ ($C_0>1$ is sufficient) and using the initial assumption $C_1\log N\leq \delta^3k$ for a sufficiently large constant $C_1>0$ we can conclude. In particular pick $C_1>4/c_6$. With these, we ensured that the total distortion is $c_{final}\delta$ for an absolute constant $c_{final}>0$ with the desired probability. Rescaling $\delta$ as $\delta\rightarrow c_{final}^{-1}\delta$, we conclude with the advertised result in Theorem \ref{main result}.\end{proof}



\small{
\bibliography{Bibfiles}
\bibliographystyle{plain}
}

\appendix
\section{Standard results}

\begin{lemma} \label{inf lip}Given vectors $\vb,\ub$, $f(\vb)=\tn{\diag{\ub}\vb}$ is $\ino{\ub}$ Lipschitz function.
\end{lemma}
\begin{lemma} [Lipschitz concentration of Gaussians] \label{lip gauss}If $f:\R^n\rightarrow \R$ is an $L$-Lipschitz function, for a standard Gaussian vector $\g$ $\Pro(|f(\g)-\E[f(\g)]|>t)\leq 2\exp(-t^2/(2L^2))$.
\end{lemma}
\begin{lemma} \label{expo bound}Let $\vb_1,\dots,\vb_k\in\R^n$ be vectors satisfying $\tn{\vb_i}\leq \ell$. Let $\Vb=[\vb_1~\dots~\vb_k]^*$ and $\g\sim\Nn(0,\Iden_n)$. Then, we have that
\beq
\Pro(\tn{\Vb\g}^2\geq \|\V\|_F^2+t)\leq \exp(-c\min\{\frac{t^2}{\|\Vb\|_F^2\|\Vb\|^2},\frac{t}{\|\Vb\|^2}\}).\nn
\eeq
\end{lemma}
\begin{proof} Let $\Vb$ have singular value decomposition $\Ub_L\Sigma\Ub_R^*$ where $\Sigma\in\R^{k\times k}$. $\Ub_R^*\g\sim\g$ and $\Ub_L$ does not affect the $\ell_2$ norm. Hence $\tn{\Ub_L\Sigma\Ub_R^*}^2\sim \tn{\Sigma\g}^2$ which is a weighted sum of subexponentials where weights are at most $\|\Vb\|^2$. Denoting the $i$th weight by $w_i=\sigma_i(\Vb)^2$ we have that
\beq
\sum_{i=1}^kw_i=\|\Vb\|_F^2,~\sup_{1\leq i\leq w_i}\leq \|\Vb\|^2.\nn
\eeq
Subject to these constraints, we are interested in finding $\sum_{i=1}^kw_i^2$. Observe that if $a>b>c$ we have that $(a+c)^2+(b-c)^2>a^2+b^2$. Consequently, without losing generality, we can assume that nonzero singular values are as large as possible, namely $\|\Vb\|$ so that there are $\frac{\|\Vb\|_F^2}{\|\Vb\|^2} $ nonzero values equal to $\|\Vb\|$.
\beq
\sum_{i=1}^kw_i^2\leq \frac{\|\Vb\|_F^2}{\|\Vb\|^2} \|\Vb\|^4=\|\Vb\|_F^2\|\Vb\|^2.\nn
\eeq
 With this bound, Proposition $5.16$ of \cite{vershynin2010introduction} yields the desired result.
\end{proof}
\begin{lemma} \label{ortho cost} Let $\ab,\bb$ be two unit vectors obeying $\max\{\ino{\ab},\ino{\bb}\}\leq \rho$. Let $\theta$ be the angle in between. Let $\bb'=\bb-\ab\ab^*\bb$. We have that $\ino{\bar{\bb}'}\leq \frac{2\rho}{\sin(\theta)}$.
\end{lemma}
\begin{proof} Clearly $\ino{\bb'}\leq \ino{\bb}+\ino{\ab^*\bb\ab}\leq 2\rho$. On the other hand since the angle between is $\theta$, $\tn{\bb'}=\sin(\theta)$.
\end{proof}


\begin{lemma} \label{angular perturb} Let $\x,\x',\y,\y'$ be unit length vectors satisfying $|\x^*\y-\x'^*\y'|\leq \alpha$. We have that $|\ang(\x,\y)-\ang(\x',\y')|\leq 5\sqrt{\alpha}$. We also have that for a unit vector $\x$ and a perturbation $\vb$, $\ang(\x,\x+\vb)\leq 5\tn{\vb}$.
\end{lemma}
\begin{proof} Without losing generality, let $\theta=\ang(\x,\y)$ and $\theta'=\ang(\x',\y)$ where $0\leq \theta\leq \theta'<\pi$. We are given that
\beq
\cos(\theta)-\cos(\theta')=\int_{\theta}^{\theta'}\sin(x)dx\leq \alpha.\nn
\eeq
Using the fact that $\sin(x)$ is increasing over $[0,\pi/2]$ and decreasing over $[\pi/2,\pi]$, we have that
\beq
2\int_0^{(\theta'-\theta)/2}\sin(x)dx\leq\int_{\theta}^{\theta'}\sin(x)dx.\nn
\eeq
If $\theta'-\theta<\pi/2$, using the fact that $\sin(x)/x$ is decreasing over $[0,\pi/2]$
\beq
2\int_0^{(\theta'-\theta)/2}\sin(x)dx\geq (2\sqrt{2}\pi^{-1})^2(\theta'-\theta)^2\geq 0.9^2(\theta'-\theta)^2.\nn
\eeq
This implies $0.9^2(\theta'-\theta)^2\leq \alpha$. Otherwise, $\alpha\geq2\int_0^{(\theta'-\theta)/2}\sin(x)dx\geq 2\int_0^{\pi/4}\sin(x)dx\geq 0.5\implies \sqrt{\alpha}\geq 0.7$. On the other hand $\theta'-\theta<\pi$ which implies $\theta'-\theta<(\pi/0.7)\sqrt{\alpha}$.
Consequently $ \theta'-\theta\leq \max\{\pi/0.7,1.2\}\sqrt{\alpha}$.

Suppose $\x$ is a unit length vector and $\x''$ be the projection of $\x$ on $\x'$. Clearly $\tn{\x'-\x}\geq \tn{\x''-\x}$ and $\ang(\x'',\x)=\ang(\x',\x)$. $\tn{\x''-\x}$ has a simple form namely it is equal to $\sin(\theta)$. Now if $\theta<\pi/4$, $\sin(\theta)>2\sqrt{2}\pi^{-1}\theta$ so that $\ang(\x',\x)\leq 2\sqrt{2}\pi^{-1}\tn{\x'-\x}$.
If $\theta>\pi/4$, $\tn{\x'-\x}\geq \sqrt{1/2}$ and $\theta\leq \pi$ which implies $\theta\leq \sqrt{2}\pi\tn{\x-\x'}$.
%
\end{proof}

\section{Results on random matrices}
\begin{lemma} \label{gauss tail}Let $\g\in\R^n$ be a standard Gaussian vector and $E$ be an event over $\g$ that holds with probability $p$. We have that
\beq
\E[\tn{\g}^2\big|E]\Pro(E)\leq  (9n+2\log p^{-1})p.\nn
\eeq
Setting $p=n^{-3}$ yields right hand side is at most $15n^{-2}$.
\end{lemma}
\begin{proof} Let $r>0$ be the number for which $\Pro(\tn{\g}\geq r)=p$ and $L$ be the associated event. Then $\Pro(E\cap \bar{L})=\Pro(\bar{E}\cap L)$ and
\beq
\E[\tn{\g}^2\big|E \cap \bar{L}]\leq \E[\tn{\g}^2\big|\bar{E}\cap L].\nn
\eeq
This implies that
\beq
\E[\tn{\g}^2\big|E]p\leq \E[\tn{\g}^2\big|L]p=\int_{\tn{\g}>r} \tn{\g}^2 d\g. \nn
\eeq
Let $a=\tn{\g}$ and $p(t)$ be the density function of $a$ and $Q(t)=\Pro(a>t)$. Using Lipschitzness of $\ell_2$ norm, we have that for $t>\sqrt{n}$, $Q(t)\leq \exp(-(t-\sqrt{n})^2)$.
\begin{align}
\E[\tn{\g}^2\big|L]p&=\int_{a>r} a^2p(a) da = -\int_{a>r} a^2dQ(a)=\int_{a>r} Q(a)da^2-[Q(a)a^2]_r^\infty \nn
&=\int_{a>r} Q(a)da^2+Q(r)r^2.\nn
\end{align}
We also have that
\beq
\int_{a>r} p(a) da=Q(r)\leq \exp(-(r-\sqrt{n})^2).\nn
\eeq
which implies $p\leq \exp(-(r-\sqrt{n})^2)\implies r\leq \sqrt{\log p^{-1}}+\sqrt{n}$. 
Construct an alternative distribution where $p'(a)=p(a)$ for $a\leq r$, $p'(a)=0$ for $r<a\leq \sqrt{\log p^{-1}}+\sqrt{n}$ and $Q'(a)=\exp(-(a-\sqrt{n})^2)$ for $a>\sqrt{\log p^{-1}}+\sqrt{n}$. This choice ensures that $Q'(a)>Q(a)$ for all $a\geq 0$ hence
\beq
\int_{a>r} a^2p'(a)da=\int_{a>r} Q'(a)da^2+Q'(r)r^2\geq \int_{a>r} Q(a)da^2+Q(r)r^2.\nn
\eeq
Consequently, we will use $Q'$ to upper bound the Gaussial tail. We have that
\beq
\int_{a>r} a^2p'(a)da=\int_{a>\sqrt{\log p^{-1}}+\sqrt{n}} a^2p'(a)da.\nn
\eeq
Finally, we need to estimate the right hand side. For $r'=\sqrt{n}+c$ and $c\geq 1$, we have that
\beq
\int_{a>r'} 2aQ'(a)da=\int_{u>c}2(\sqrt{n}+u)\exp(-u^2/2)du=7\sqrt{n}\exp(-c^2/2).\nn
\eeq
We also have the estimate $Q(r')r'^2\leq 2(n+c^2)\exp(-c^2/2)$. Setting $c=\sqrt{\log p^{-1}}$ and $p=n^{-3}$ we find that
\beq
\E[\tn{\g}^2\big|L]p\leq (9n+2\log p^{-1})p\leq 15n^{-2}.\nn
\eeq
\end{proof}

\begin{lemma} [Infinity norm of random modulation]\label{rand modulation} Let $\{\vb_i\}_{i=1}^N$ be a finite set of points. Let $\bb\in\R^n$ be a vector with independent Rademacher entries and let $\Ub\in\R^n$ be the unitary Hadamard matrix where entries are $\pm \sqrt{1/n}$. Let $\w_i=\Ub\diag{\bb}\vb_i$. With probability $1-\exp(-c_0\log N)$, for all $1\leq i\leq N$, we have that
\beq
\sup_{1\leq i\leq N}\ino{\w_i}\leq \frac{\sqrt{\log n}+\sqrt{\log N}}{\sqrt{n}}.\nn
\eeq
\end{lemma}
\begin{proof} Observe that $\w_i=\Ub\diag{\vb_i}\bb$ hence each entry of $\w_i$ is a weighted linear combination of subgaussians where the weights are $\Ub_{jk}\vb_{ki}=\pm \vb_{ki}/\sqrt{n}$. In particular $\sum_{k=1}^n|\Ub_{jk}\vb_{ki}|^2=1/n$ hence $\w_{ij}$ has $\order{1/\sqrt{n}}$ subgaussian norm. Consequently for any $1\leq j\leq n$
\beq
\Pro(|\w_{ij}|\geq t)\leq \exp(-cnt^2/2).\nn
\eeq
Pick $t=c'(\frac{\log n+\log N}{{n}})^{1/2}$ and apply a union bound over all $1\leq j\leq n$ and all $1\leq i\leq N$ to conclude.
\end{proof}

\begin{lemma} [Embedding most vectors]\label{most modulation} Let $\{\vb_i\}_{i=1}^N$ be a finite set of points. Let $\bb\in\R^n$ be a vector with independent Rademacher entries and let $\Ub\in\R^n$ be the unitary Hadamard matrix where entries are $\pm \sqrt{1/n}$. Let $\w_i=\Ub\diag{\bb}\vb_i$. With probability $1-p$, for at least $(1-cp^{-1}n^{-3})N$ points $\w_i$ ($1\leq i\leq N$), we have that
\beq
\sup_{1\leq i\leq N}\ino{\w_i}\leq C\frac{\sqrt{\log n}}{\sqrt{n}}.\nn
\eeq
\end{lemma}
\begin{proof} Observe that $\w_i=\Ub\diag{\vb_i}\bb$ hence each entry of $\w_i$ is a weighted linear combination of subgaussians where the weights are $\Ub_{jk}\vb_{ki}=\pm \vb_{ki}/\sqrt{n}$. In particular $\sum_{k=1}^n|\Ub_{jk}\vb_{ki}|^2=1/n$ hence $\w_{ij}$ has $\order{1/\sqrt{n}}$ subgaussian norm. Consequently for any $1\leq j\leq n$
\beq
\Pro(|\w_{ij}|\geq t)\leq \exp(-cnt^2).\nn
\eeq
Pick $t=c'(\frac{\log (n)}{{n}})^{1/2}$ to ensure that $\Pro(|\w_{ij}|\geq t)\leq c''n^{-4}$. Applying a union bound over the entries, this ensures $\Pro(\ino{\w_{i}}\geq t)\leq c''n^{-3}$. Let $N_s$ be the number of $\w_i$ obeying the bound $\ino{\w_i}\leq t$. We have that
\[
\E[N_s]\geq 1-c''n^{-3}.
\] 
Hence $N-N_s$ is a nonnegative random variable obeying $\E[N-N_s]\leq c''n^{-3}$. Applying Markov's inequality $\Pro(N-N_s>p^{-1}c''n^{-3})\leq p$.
\end{proof}

\begin{lemma} \label{lemma simple but useful}Let $\A$ be a random matrix with unit length columns. Suppose $\E\|\A^*\A-\Iden\|\leq \alpha$. Let $\B=\A\diag{\alpha}$ where $\alpha$ is a diagonal matrix whose entries lie between $\sqrt{1\pm \eps}$ and $\alpha$ is allowed to depend on $\A$. We have that
\beq
\E\|\B^*\B-\Iden\|\leq2\alpha+\eps.\nn
\eeq
\end{lemma}
\begin{proof} Let $\phi=\|\A^*\A-\Iden\|$. We have that $\|\B\|^2\leq (1+\phi)(1+\eps)$ and $\sigma_{\min}(\B)^2\geq (1-\phi)(1-\eps)$. Consequently
\beq
\E[\|\B\|^2-1]\leq \E[\phi]+\eps+\E[\phi\eps],~1-\sigma_{\min}(\B)^2\leq \E[\phi]+\eps-\E[\phi\eps].\nn
\eeq
\end{proof}
\section{\large{Generalizations of Tropp's ``Incoherent Subdictionary Theorem''}}
\begin{definition} $\Phi\in\R^{m\times n}$ is called a dictionary with coherence $\mu$ if columns of $\Phi$ have unit length and coherence is defined as $\mu=\sup_{i\neq j}|\phi_i^*\phi_j|$ where $\phi_i$ is the $i$th column of $\Phi$.
\end{definition}
\begin{definition}[Restriction] $\Rb\in\R^{m\times m_1}$ is called a restriction operator if $\A\Rb\in\R^{n\times m_1}$ is a matrix obtained by selecting $m_1$ columns of $\A$ for any $\A\in\R^{n\times m}$ and any $n\geq 1$. If $\Rb$ select $m_1$ columns uniformly at random, we shall call it random restriction. A random subdictionary of $\A$ is obtained by applying the restriction $\Rb$ to get $\A\Rb$.
\end{definition}
Define $\|\cdot\|_{1,2}$ norm of a matrix to be the largest $\ell_2$ norm of its columns. The next result will be beneficial for the derivation.
\begin{theorem}[Theorem $8$ of \cite{tropp2008conditioning}]\label{single restrict} Let $\A$ be a matrix with $N$ columns and let $\Rb$ be a restriction to $m$ coordinates chosen uniformly at random. Fix $q\geq 1$. For any $p\geq \max\{2,2\log (\text{rank}\A\Rb^*),q/2\}$ we have that
\beq
(\E\|\A\Rb^*\|^q)^{1/q}\leq 3\sqrt{p}\|\A\|_{1,2}+\sqrt{m/N}\|\A\|.\nn
\eeq
Observe that $(a+b)^q\leq (2\max(a+b))^q\leq (2a)^q+(2b)^q$ hence
\beq
\E\|\A\Rb^*\|^q\leq (6\sqrt{p}\|\A\|_{1,2})^q+(2\sqrt{m/N}\|\A\|)^q.\nn
\eeq
\end{theorem}
The following is our variation of Tropp's spectral norm bounds on incoherent subdictionaries.
\begin{theorem} \label{tropp var}Suppose $\Phi\in\R^{n\times 2n}$ is a random matrix such that all of its realizations are incoherent dictionaries with coherence $\mu$. Pick a random subdictionary $\X\in\R^{n\times 2k}$ of $\Phi$. Define the function $f(\Rb)=\E_{\Phi}\|\X^*\X-\Iden\|$. For $u\geq \sqrt{2\log n+1}$, we have that
\beq
\Pro_{\Rb}(f(\Rb)\geq c'(u\sqrt{m(\log m+1)}\mu+ \frac{m}{n}\|\Phi\|^2))\leq \exp(-u^2/4).\label{rand var}
\eeq
\end{theorem}
\begin{proof} The proof exactly follows the work by Tropp, namely Section $6$ of \cite{tropp2008conditioning}. We will only point out the main differences as almost all of the argument overlaps. Let $\Rb$ be the random restriction for which $\X=\Phi\Rb$. We first establish the following result.
\begin{theorem}\label{pth moment} For $q\geq 2\log n+1$, we have that
\beq
(\E_{\Rb}(\E_{\Phi}\|\X^*\X-\Iden\|)^q)^{1/q}\leq  (\E_{\Rb,\Phi}\|\X^*\X-\Iden\|^q)^{1/q}\leq c(\sqrt{qm(\log m+1)}\mu+ \frac{m}{n}\|\Phi\|^2).\nn
\eeq
\end{theorem}
\begin{proof} For the sake of completeness, we repeat most of the arguments in \cite{tropp2008conditioning}. First note that $\X^*\X-\Iden=\Rb\Hb\Rb^*$ where $\Hb=\Phi^*\Phi-\Iden$. A standard symmetrization argument (Theorem $9$ of \cite{tropp2008conditioning}) ensures that there exists a submatrix $\hat\Hb\in\R^{n/2\times n/2}$ (where columns and rows correspond to disjoint subsets) and restrictions $\Rb_1,\Rb_2$ such that
\beq
(\E_{\Rb}\|\X^*\X-\Iden\|^q)^{1/q}\leq 2(\max_{m_1+m_2=m}\E_{\Rb_1,\Rb_2}\|\Rb_1\hat\Hb\Rb_2\|^q)^{1/q}.\label{sym part}
\eeq
Exponentiating both sides, this implies
\beq
\E_{\Rb}\|\X^*\X-\Iden\|^q\leq 2^q\max_{m_1+m_2=m}\E_{\Rb_1,\Rb_2}\|\Rb_1\hat\Hb\Rb_2\|^q.\nn
\eeq
Hence, we shall upper bound the right-hand side. This will be done in two steps by first taking expectation over $\Rb_2$ and then $\Rb_1$.
\beq
\E_{\Rb_1,\Rb_2}\|\Rb_1\hat\Hb\Rb_2\|^q=\E_{\Rb_1}\E_{\Rb_2}\|\Rb_1\hat\Hb\Rb_2\|^q.\nn
\eeq
Applying Theorem \ref{single restrict} with $p=\max\{2,2\log(m/2)+1,q/2\}$ we have that
\beq
\E_{\Rb_2}\|\Rb_1\hat\Hb\Rb_2\|^q\leq (6\sqrt{p}\|\Rb_1\hat\Hb\|_{1,2})^q+(\sqrt{8m_2/n}\|\Rb_1\hat\Hb\|)^q.\nn
\eeq
Similar to Tropp, the coherence assumption ensures that $\|\Rb_1\hat\Hb\|_{1,2}\leq \mu\sqrt{m}$ to obtain
\beq
\E_{\Rb_2}\|\Rb_1\hat\Hb\Rb_2\|^q\leq (6\sqrt{p}\mu\sqrt{m})^q+(\sqrt{8m_2/n}\|\Rb_1\hat\Hb\|)^q.\nn
\eeq
The remaining task is to upper bound $\E\|\Rb_1\hat\Hb\|^q$. Reapplying Theorem \ref{single restrict}, we have that
\beq
\E\|\Rb_1\hat\Hb\|^q=\E\|\hat\Hb^*\Rb_1^*\|^q\leq (6\sqrt{p}\mu\sqrt{n})^q+(\sqrt{8m_1/n}\|\hat\Hb\|)^q.\nn
\eeq
The combination of the last two inequalities, yields that, for any $\Phi$ obeying the coherence and spectral norm bounds, we have
\begin{align}
\E_{\Rb_1,\Rb_2}\|\Rb_1\hat\Hb\Rb_2\|^q&\leq (6\sqrt{p}\mu\sqrt{m})^q+(\sqrt{8m_2/n}6\sqrt{p}\mu\sqrt{n})^q+(\sqrt{8m_2/n}\sqrt{8m_1/n}\|\hat\Hb\|)^q\nn\\
&\leq (c_1\sqrt{pm}\mu)^q+(c_2 \frac{m}{n}\|\hat\Hb\|)^q.\nn
\end{align}
Since this holds for all realizations of $\Phi$, we can take an additional expectation over $\Phi$ to conclude
\beq
\E_{\Rb,\Phi}\|\X^*\X-\Iden\|^q\leq (c_1\sqrt{pm}\mu)^q+(c_2 \frac{m}{n}\|\hat\Hb\|)^q\implies (\E_{\Rb,\Phi}\|\X^*\X-\Iden\|^q)^{1/q}\leq c(\sqrt{pm}\mu+ \frac{m}{n}\|\hat\Hb\|).\nn
\eeq
For $q\geq 1$, this also implies that
\beq
(\E_{\Rb}(\E_{\Phi}\|\X^*\X-\Iden\|)^q)^{1/q}\leq  (\E_{\Rb,\Phi}\|\X^*\X-\Iden\|^q)^{1/q}\leq c(\sqrt{pm}\mu+ \frac{m}{n}\|\hat\Hb\|).\nn
\eeq
Picking $p=q(\log m/2+1)$, and using the estimate $\|\hat\Hb\|\leq \|\Hb\|\leq \|\Phi^*\Phi\|+1\leq 2\|\Phi\|^2$ we obtain that
\beq
(\E_{\Rb}(\E_{\Phi}\|\X^*\X-\Iden\|)^q)^{1/q}\leq  (\E_{\Rb,\Phi}\|\X^*\X-\Iden\|^q)^{1/q}\leq c(\sqrt{qm(\log m+1)}\mu+ \frac{m}{n}\|\Phi\|^2).\nn
\eeq
\end{proof}
Now letting $f(\Rb)=\E_{\Phi}\|\X^*\X-\Iden\|$ and applying Proposition $10$ of \cite{tropp2008conditioning}, for $u\geq 1$,  we obtain that
\beq
\Pro(f(\Rb)\geq c'(u\sqrt{m(\log m+1)}\mu+ \frac{m}{n}\|\Phi\|^2))\leq \exp(-u^2/4)\nn
\eeq
which is the desired concentration bound.
\end{proof}

\subsection{Asymmetric version of Tropp's incoherent subdictionary result}
We first prove the following variation of Theorem $25$ of \cite{tropp2008conditioning}. This result assumes the matrix to have even dimensions but the odd case can be shown with minimal modification of the proof strategy. The proof exactly follows the argument of Tropp however we will provide it here for the sake of completeness. We remark that Tropp's result was based on more classical results due to Bourgain and Tzafriri \cite{bourgain1987invertibility,ledoux2013probability}.
\begin{theorem} \label{thm asym}Let $\A$ be a $2n\times 2n$ matrix with a $0$ diagonal. Let $\Rb$ be a restriction to $m$ random coordinates. Fix $q\geq 1$. There exists a partition of the coordinates $\{1,2,\dots,2n\}$ into two blocks $T_1$ and $T_2$ with $N$ elements each so that
\[
\E\|\Rb\A\Rb^*\|^q\leq 2^q\max_{m_1+m_2=m}\E\|\Rb_1\A_{T_1\times T_2}\Rb_2\|^q.
\]
\end{theorem}
\begin{proof} We prove the result for $q=1$. Identical argument applies to $q>1$ case. Let $a_{ij}$ denote the $ij$th coordinate of the matrix $\A$ and $\e_i$ denote the $i$th vector of the standard basis. Define the matrices $\B_{jk}=a_{jk}\e_j\e_k^*$. $\delta\in\{0,1\}^{2n}$, be a vector which has exactly $m$ components equal to $1$. Then, we have
\[
\Rb\A\Rb^*=\sum_{j\neq k}\delta_j\delta_k\B_{jk}
\]
We wish to bound the expectation
\[
M=\E_{\delta}\|\sum_{j\neq k}\delta_j\delta_k\B_{jk}\|
\]
Let $\eta\in\R^{2n}$ be a random vector with exactly $n$ coordinates equal to $1$. For $j\neq k$, we have that
\beq
\E_{\eta}[\eta_j(1-\eta_k)+\eta_k(1-\eta_j)]=\frac{n}{2n-1}.
\eeq
Now, based on these, applying Jensen's inequality, we have the following list of inequalities
\begin{align}
M&=\frac{2n-1}{n}\E_\delta\|\sum_{j\neq k}\E_\eta[\eta_j(1-\eta_k)+\eta_k(1-\eta_j)]\delta_j\delta_k\B_{jk}\|\nn\\
&<2\E_\delta\E_\eta\|\sum_{j\neq k}[\eta_j(1-\eta_k)+\eta_k(1-\eta_j)]\delta_j\delta_k\B_{jk}\| \nn\\
&<2[\E_\delta\E_\eta\|\sum_{j\neq k}\eta_j(1-\eta_k)\delta_j\delta_k\B_{jk}\|+\E_\delta\E_\eta\|\sum_{j\neq k}\eta_k(1-\eta_j)\delta_j\delta_k\B_{jk}\|]\nn\\
&<4\E_\delta\E_\eta\|\sum_{j\neq k}\eta_j(1-\eta_k)\delta_j\delta_k\B_{jk}\|\nn
\end{align}
It follows that there exists a $0-1$ vector $\eta^*$ containing exactly $n$ $1$s such that
\[
M<4\E_\delta\|\sum_{j\neq k}\eta^*_j(1-\eta^*_k)\delta_j\delta_k\B_{jk}\|
\]
Note that this vector $\eta^*$ partitions the set $\{1,2,\dots,2n\}$ into two parts $T_1,T_2$ each containing $N$ elements. $T_1$ corresponds to the coordinates $\eta_i=1$ and $T_2$ corresponds to the coordinates $\eta_i=0$. Calling these parts $T_1,T_2$, we can rewrite the inequality as
\[
M<4\E_\delta\|\sum_{j\in T_1,k\in T_2}\delta_j\delta_k\B_{jk}\|
\]
Next, let number of active coordinates of $\delta$ over $T_1$ be $m_1$. Observe that conditioned on the choice of $T_1$ and $m_1$, the $m_1$ and $m-m_1$ active coordinates of $\delta$ are distributed uniformly at random over $T_1$ and $T_2$. This is due to the fact that $\delta$ is independent of $\eta^*$. With this, the inequality takes the advertised form
\[
M<4\E_{m_1,m_2}\|\Rb_1^*\A_{T_1\times T_2}\Rb_2\|\leq \max_{m_1+m_2=m}4\E\|\Rb_1^*\A_{T_1\times T_2}\Rb_2\|
\]
\end{proof}
Using Theorem \ref{thm asym} and repeating the proof of Theorem \ref{tropp var} line by line we can conclude with the following result.
\begin{theorem} \label{tropp var2}Suppose $\Phi\in\R^{n\times 2n}$ is a random matrix such that all of its realizations are incoherent dictionaries with coherence $\mu$. Pick a random subdictionary $\X_1\in\R^{n\times k}$ from first $n$ columns of $\Phi$. Pick the same $k$ coordinates from the second $n$ columns of $\Phi$ and form $\X_2$. Define the function $f(\Rb)=\E_{\Phi}\|\X_1^*\X_2\|$. For $u\geq \sqrt{2\log n+1}$, we have that
\beq
\Pro_{\Rb}(f(\Rb)\geq c'(u\sqrt{m(\log m+1)}\mu+ \frac{m}{n}\|\Phi\|^2))\leq \exp(-u^2/4).\label{rand var}
\eeq
\end{theorem}
\begin{proof} Denote the first and second $n$ columns of $\Phi$ by $\Phi_1$ and $\Phi_2$ respectively and set $\Theta=\Phi_1^*\Phi_2$. Let $\Rb\in\R^{n\times k}$ be a random restriction. Observe that
\[
\E_{\Phi}\|\X_1^*\X_2\|=\E_{\Phi,\Rb}\|\Rb^*\Phi_1^*\Phi_2\Rb\|=\E_{\Phi,\Rb}\|\Rb^*\Theta\Rb\|
\]
Next, split the diagonal and off-diagonal entries of $\Theta$ and apply the argument in Theorem \ref{tropp var} where we replace the inequality \eqref{sym part} with the estimate of Theorem \ref{thm asym}.
\end{proof}
\end{document}